\documentclass[10pt,conference]{IEEEtran}
\usepackage{amsmath}
\usepackage{amsfonts}
\newcommand{\iexpe}[1]{\langle#1\rangle}

\newcommand{\iintx}[1]{\int_{\mathcal{X}}#1\mathrm{d}x}
\newcommand{\isumx}{\sum_{x \in \mathcal{X}}}
\newcommand{\isumm}{\sum_{m=1}^M}
\newcommand{\isumn}{\sum_{i=1}^n}
\newcommand{\ib}[1]{\left(#1\right)}
\newcommand{\iz}{\widehat{Z}}
\newcommand{\ic}{\mathcal{C}}
\newcommand{\ix}{\mathcal{X}}

\newcommand{\ig}{\Gamma}
\newcommand{\id}{\enspace.}
\newcommand{\ico}{\enspace,}

\newcommand{\idiff}[2]{\frac{\partial\; #1}{\partial #2}}
\newcommand{\imin}[1]{\arg\min\limits_{#1}}
\newcommand{\iln}[2]{\ln_q \frac{#1}{#2}}

\newtheorem{thm}{Proposition}

\begin{document}
\title{On Shore and Johnson properties for a Special Case of Csisz\'{a}r
$f$-divergences}
\author{Jithin Vachery and Ambedkar Dukkipati\\
Dept of CSA, IISc\\
\{jithinvachery, ambedkar\} @csa.iisc.ernet.in
}
\date{}
\maketitle
\thispagestyle{empty}
\bibliographystyle{unsrt}

\begin{abstract}
The importance of power-law distributions is attributed to
the fact that most of the naturally occurring phenomenon exhibit this
distribution. While exponential distributions can be derived by minimizing
KL-divergence w.r.t some moment constraints, some power law distributions can be
derived by minimizing some
generalizations of KL-divergence (more specifically some special cases of
Csisz\'{a}r $f$-divergences).
Divergence minimization is very well studied in information theoretical
approaches to statistics.
In this work we study properties of minimization of
Tsallis divergence, which is a special case of Csisz\'{a}r $f$-divergence. In
line with the work by Shore and Johnson
(IEEE Trans. IT, 1981), we  examine the properties exhibited
by these minimization methods including the Pythagorean property.
\end{abstract}

\section{INTRODUCTION}
Shannon measure of information, also called entropy, is central to information
theory which has wide range of applications spanning, communication theory,
statistical mechanics, probability theory, statistical inference etc.
\cite{CoverThomas:1991:ElementsOfInformationTheory}. It quantifies uncertainty
or information that is associated with a discrete random variable by taking an
average of uncertainty (Hartley information) associated with each state. The
first generalization of this measure of information was suggested by
R$\acute{e}$nyi \cite{Renyi:1965:OnTheFoundationsOfInformationTheory}. He
replaced the linear averaging by K-N averages (Kolmogrov-Nagumo averages) and
imposed additivity constraint. Havrda and Charvat
\cite{HavrdaCharvat:1967:QuantificationMethodOfClassificationProcess} introduced
one more generalization which is now
known as nonextensive entropy or Tsallis entropy
\cite{TatsuakiTakeshi:2001:WhenNonextensiveEntropyBecomesExtensive,
Suyari:2004:GeneralizationOfShannonKhinchinAxioms,
Furuichi:2005:OnUniquenessTheormForTsallisEntropyAndTsallisRelativeEntropy},
which has been studied in statistical mechanics.\\
Another important notion is that of finding the distance or divergence between
two probability distributions. The information measure capturing this is
KL-divergence, which is the directed distance between two probability
distributions. KL-divergence is a special case of Tsallis divergence, which in
turn is a special case of Csisz\'{a}r $f$-divergence
\cite{CsiszarShields:2004:InformationTheoryAndStatistics}. KL-divergence plays a
central role in Kullback's minimum divergence principle, Which is a means
of estimating the probability distribution of a system. It  suggests the
minimization of KL-divergence using a given prior distribution, subject to
moment constraints as the estimation technique. Kullback's minimum divergence
principle reduces to Jaynes maximum entropy principle when we use uniform
distribution as the prior. Kullback's minimum divergence principle can be
extended to generalized divergences. When applied to classical KL-divergence,
this yields a distribution from the exponential family. Whereas applying
Kullback's principle  to Tsallis divergence gives a power-law distribution.\\
Exponential distributions are very important class of distributions and many
problems have been successfully modeled using this \cite{clark:2004:primer}.
Though exponential distributions are used in many modeling problems
\cite{bishop2006pattern} due to theoretical tractability, many naturally
occurring phenomena  exhibit power-law distributions. It is of great practical
and theoretical interest to study both these family of distributions.\\
In this work we have been able to establish many properties for Tsallis
divergence. We have established the property of \textit{transformation
invariance} and \textit{subset independence}. In addition we have found some
properties for Tsallis divergence minimization in classical constraints viz.
\textit{uniqueness}, \textit{reflexiveness}, \textit{idempotence},
\textit{invariance}, \textit {weak subset independence} and \textit{subset
aggregation}. In this work we have also attempted to derive a
\textit{Pythagorean property}. In addition we have proposed a $q \leftrightarrow
2-q$ \textit{additive transformation} for Tsallis divergence. \\
The paper is organized as follows. In Section \ref{sec2} we introduce the
preliminaries and basics required for understanding the results. Sections 
\ref{sec3} through \ref{snjmini} are dedicated to the results and observations
made. In
these sections we perform Tsallis divergence minimization for classical
constraints and we follow it up with the analysis of the properties exhibited by
the same. In particular we are study about the Shore and Johnson properties. In
the subsequent section we discuss about the a transformation relation which we
established.

\section{Preliminaries and Background}\label{sec2}
\subsection{Exponential family and KL Divergence}
In many of the problems we might have a prior estimate of the probability
distribution and given such a prior we are interested in finding
the probability distribution that is closest to this prior, which also satisfies
the set of linear constraints. To define the notion of closeness we need a
distance measure between two distributions. One such distance measure is KL
divergence \cite{Kullback:1959:InformationTheoryAndStatistics} defined as
\[
I(p||r) = \isumx p(x) \ln \left(\frac{p(x)}{r(x)}\right)\ico
\]
where $r$ is the prior. The minimization of KL-divergence results in a posterior
which is from the exponential family.
\subsection{Power-Law distribution and Generalized Divergence}
$f$-divergence is a generalized measure of divergence, that was
introduced by Csis$\acute{z}$ar
 \cite{CsiszarShields:2004:InformationTheoryAndStatistics} and independently
by Ali \& Silvey \cite{Ali_Silvey:1966:AGeneralClassOfCoefficientsOfDivergence}.
Let $f (t)$ be a real valued convex function defined for $t > 0$, with $f (1) =
0$. The $f$-divergence of a distribution $p$ from $r$ is defined by
\[
D_f(p||r) = \isumx  r(x) f\left(\frac{p(x)}{r(x)}\right)\id
\]
Here we take $0f\left(\frac{0}{0}\right) = 0, f(0) = \lim_{t \to 0} f(t).$
$f$-divergence has many important properties like non-negativity, monotonicity
and convexity. This has been used in many applications like speech recognition
\cite{qiao2010study}, analysis of contingency tables 
\cite{CsiszarShields:2004:InformationTheoryAndStatistics}, etc. By specializing
$f$ to various functions we get different divergences like KL-divergence,
$\chi^2$-divergence, Hellinger distance, variational distance,
Tsallis-divergence, etc. On setting $f(t)= t\ln_qt$ we get Tsallis divergence
\cite{TatsuakiTakeshi:2001:WhenNonextensiveEntropyBecomesExtensive}, defined as 
\begin{equation*}
I_q(p||r) = - \isumx  p(x) \ln_q\frac{r(x)}{p(x)}\ico
\end{equation*}
where $\ln_q$ is $q$-logarithm function
\cite{Borges:2004:ApossibleDeformedAlgebra}, defined as, $\ln_qx =
\frac{x^{(1-q)}-1}{1-q} \quad (x>0, q \in \mathbb{R})$.
Tsallis divergence recovers KL-divergence for $q \rightarrow 1$ i.e., $
\lim_{q \to 1} I_q(p||r) = I(p||r)$.
For values of $q>0$ we have $I_q(p||r) \ge 0 $ and Tsallis divergence becomes
a convex function of both the parameters. Tsallis divergence also exhibits
pseudo
additivity property, i.e., $ I_q(X1\times X2||Y1\times Y2) = I_q(X1||X2)
\oplus_q
I_q(Y1||Y2)$,
where $X1$ and $X2$ are independent, so are $Y1$ and $Y2$. Here $\oplus_q$ is
addition in $q$-deformed algebra \cite{Borges:2004:ApossibleDeformedAlgebra}
defined as, $x\oplus_qy = x+y+(1-q)xy$. In the minimization of Tsallis
divergence the choice of constraints play an important role
\cite{TsallisMendesPlastino:1998:TheRoleOfConstraints}.\\
Tsallis Divergence minimization with respect to $q$-expectation constraint has
been studied by
\cite{
BorlandPlastinoTsallis:1998:InformationGainWithinNonextensiveThermostatistics}.
In this case Pythagoras theorem is established by
\cite{DukkipatiMurtyBhatnagar:2006:NonextensiveTriangleEquality,
Dukkipati:2007:NonextensivePythagorasTheorem,
DukkipatiMurtyBhatnagar:2005:PropertiesOfKullback-LeiblerCrossEntropyMinimization}
and proved in differential geometric setup by Ohara
\cite{ohara2007geometry}.\\
Tsallis divergence minimization with normalized constraints gives probability
distribution which is self referential in nature, i.e., $p(x)$ depends of
$p(x)$. Here too we have nonextensive Pythagoras property 
\cite{DukkipatiMurtyBhatnagar:2006:NonextensiveTriangleEquality,
Dukkipati:2007:NonextensivePythagorasTheorem} exhibited by Tsallis-divergence.\\
In this paper we are going to study this minimization with respect to classical
expectations, as it has the important property of convexity, ensuring a unique
solution.
\section{Basic Shore and Johnson Properties}\label{sec3}
Shore and Johnson \cite{Shore:1981:PropertiesOfCrossEntropyMinimization} in
their
work in 1981 had discussed many of the important properties of KL-divergence
minimization. We have found that many of those properties hold in the case of
Tsallis divergence. In this section we shall discuss about the properties that
pertain to Tsallis divergence, i.e., regardless of minimization.\\
In this section and section \ref{snjmini} we shall be using the following
notation.\\
Let $p$ be a pmf. on random variable $X$ taking values from $\mathcal{X}$. We
would like to impose the following linear equality and inequality constraints on
it.
\begin{align}
 \isumx p(x) &= 1 \label{const1}\ico\\
 \isumx p(x)u_m &= \iexpe{u_m}\quad m= 1\dots M \label{const2}\ico\\
 \isumx p(x)w_n &\ge \iexpe{w	_n}\quad n= 1\dots N \id\label{const3}
 \end{align}
Equations \eqref{const1},\eqref{const2} and \eqref{const3} constitute the
constraint set. This can also be considered as the information available about
the probability distribution. We shall denote a constraint set by $\ic$, and a
subscript to distinguish between different constraint sets.\\
Hence the task of divergence minimization can be viewed as,  given a
prior probability distribution $q(x)$ and constraint set $\ic$ finding  the
probability distribution $p_{min}$ such that $p_{min} = \imin{p\in\ic}
I_q(p||r)$.
It can be easily verified that the constraint set $\ic$ constitutes a convex
set.
We would like to inform that some of these notation have been borrowed from
 \cite{Shore:1981:PropertiesOfCrossEntropyMinimization}.\\
Invariance of KL-divergence to coordinate transformations enables us to
generalize KL-divergence to continious random variables. We have observed that
the invariance property holds true in the case of Tsallis divergence too.
\begin{thm}[Invariance]
 Let $\Gamma$ be a coordinate transformation from $x\in \ix$ to $y\in
\ix'$ with $(\Gamma p)(y) = J^{-1}p(x)$, where $J$ is the Jacobian $J =
\partial(y)/\partial(x)$. Let $\Gamma\mathcal{X}$ be the set of densities
$\Gamma p$ corresponding to densities $p \in \mathcal{X}$. Let $(\Gamma\ic)
\subseteq (\Gamma\ix)$ correspond to $\ic \subseteq \ix$. Then, given a prior
distribution $r$
\begin{align}
 \imin{p \in \ig\ic} I_q(p||\ig r) &= \imin{s
\in\ic}I_q(s||r)\label{res1_1}\ico\\
 \text{and }
 I_q(\ig p_{min}||\ig r) &= I_q(p_{min}||r)\label{res1_2}\ico
\end{align}
hold. where $\ig p_{min} = \imin{p \in \ig\ic} I_q(p||\ig r)$ and $p_{min} =
\imin{s \in\ic}I_q(s||r)$.
\end{thm}
\begin{proof}
 We have $(\Gamma p)(y) = J^{-1}p(x)$, where $J$ is the Jacobian $J =
\partial(y)/\partial(x)$.
\begin{align*}
 I_q(\ig p|| \ig r) &= -\int_{\ig \ix} \ig p(y) \;\ln_q \left(\frac{\ig
r(y)}{\ig
p(y)}\right) \mathrm{d}y\\
 &= -\iintx{J^{-1} p(x) \; \ln_q \left(\frac{J^{-1} r(x)}{J^{-1} p(x}\right)
\;J}\\
 &= -\iintx{p(x)\;\ln_q\frac{r(x)}{p(x)}}\\
 &= I_q(p||r)
\end{align*}
This proves \eqref{res1_2}. From \eqref{res1_2} it also follows  that the
minimum in $\ig \ic$ corresponds
to the minimum in $\ic$, which proves \eqref{res1_1}.
\end{proof}

\begin{thm}[Subset Independence]\label{res2}
 Let $S_1,S_2,\dots,S_n$ be a partition of $\ix$. Let the new information $\ic$
comprise about each of the conditional densities $p(x/x\in s_i),\; i= 1 \dots
n$. Thus, $\ic = \ic_1 \wedge \ic_2 \wedge \dots \wedge \ic_n$, where $\ic_i$
is the constraint set on the conditional densities of $S_i$. Let $\mathcal{M}$
be the new information giving the probability of being in each of the $n$
subsets, which is the constraint
\begin{equation*}
 \sum_{x\in S_i} p(x) = m_i,\quad i=1\dots n\ico
\end{equation*}
where $m_i$ are known values. Then given the prior distribution $r$,
\begin{equation}
 p_{\mathcal{CM}}^{min}(x/x\in S_i) = \imin{p \in \ic_i}
I_q(p||r_i), \quad q\in(0,1)\label{res2_1}\ico
\end{equation}
 and
\begin{align}
 I_q(p_{\mathcal{CM}}^{min}||r) &= \isumn m_i\; I_q(p_i||r_i) - \isumn m_i
\;\ln_q \frac{s_i}{m_i} \nonumber\\
 &+ (1-q) \isumn \left( m_i\; \ln_q \frac{s_i}{m_i}\;
I_q(p_i||r_i)\right)\label{res2_2}
\end{align}
hold, where
\begin{align*}
 p_{\mathcal{CM}}^{min} &= \imin{p\; \in \ic \wedge \mathcal{M}} I_q(p||r)\ico\\
 p_i(x) &= p_{\mathcal{CM}}^{min}(x/x \in S_i)\ico\\
 r_i(x) &= r(x/x \in S_i)\ico
\end{align*}
and $s_i$ are the prior probability of being in each subset, given
by $s_i = \sum_{x \in S_i} r(x)$.
\end{thm}
\begin{proof}
 \begin{equation*}
  I_q(p_{\mathcal{CM}}^{min}||r) = - \isumn \sum_{x\in S_i} m_i p_i(x)
\;\ln_q \frac{s_ir_i(x)}{m_i
r_i(x)}\id
 \end{equation*}
  Using the relation $\ln_q(xy) = \ln_qx + \ln_qy+(1-q)\ln_qx\;\ln_qy$, we
get
 \begin{align*}
  I_q(p_{\mathcal{CM}}^{min}||r) &= - \isumn \sum_{x\in S_i} m_i p_i(x) 
  \left(
  \iln{s_i}{m_i} + \iln{r_i(x)}{p_i(x)} \right.\\
  & \qquad \left.+ (1-q)\;\iln{s_i}{m_i}\;\iln{r_i(x)}{p_i(x)}
  \right)\\
    &= \isumn m_i I_q(p_i||r_i) - \isumn m_i\;\iln{s_i}{m_i}\\
  &\qquad + (1-q)\isumn \left( m_i\;\iln{s_i}{m_i} I_q(p_i||r_i)\right)\ico
 \end{align*}
 this proves \eqref{res2_2}. To prove \eqref{res2_1} it may be noted that each
of the terms
$m_i\;\iln{s_i}{m_i}$ is a constant. Hence minimizing rhs of \eqref{res2_2} is
independent of the values taken by it. i.e for $q\in (0,1)$ minimizing
$I_q(p_{\mathcal{CM}}^{min}||r)$ is equivalent to minimizing each of the terms,
$I_q(p_i||r_i)$.
\end{proof}
Let us further analyze equation \eqref{res2_1} and try to interpret it. What
this means is that, given a system which naturally partitions into subsets, we
can find the posterior densities in two different ways
\begin{enumerate}
 \item We can find the posterior $p_{\mathcal{CM}}^{min}$ and condition it on
the different subsets $S_i$ or
 \item We can condition the prior $r$ on the different subsets $S_i$ and use
that as a prior to minimize in the constraint set $\ic_i$
\end{enumerate}
By \eqref{res2_1} both these approaches should give the same result.
\section{Tsallis Divergence Minimization - Classical}\label{sec4}
The task of
minimization can be defined as follows: Minimize $I_q(p||r)$ subject to the
constraints
\begin{align}
\isumx{p(x)} &= 1\label{pxsum}\ico\\
p(x) &\ge 0\nonumber\ico\\
\isumx{u_m(x) p(x)} &= \iexpe{u_m}, \quad m = 1,\dots,M\nonumber\id
\end{align}
By choosing the Lagrangian for the minimization problem as
\begin{eqnarray*}
 \lefteqn{\mathcal{L} =}\\
 &\isumx p(x)\frac{\left[\frac{p(x)}{r(x)}\right]^{q-1}-1}{q-1}
 - \ib{\frac{q\lambda-1}{q-1}} \bigl( \isumx{p(x)} - 1\bigr) \\
 &- \isumm q\lambda\beta_m(\isumx{u_m(x) p(x)} - \iexpe{u_m})\id
\end{eqnarray*}
The distribution that we get after minimization is
\begin{equation}\label{px1}
 p(x) = r(x) \left[\lambda\Bigl(1+(q-1)\isumm \beta_m
u_m(x)\Bigr)\right]^{\frac{1}{q-1}}\id
\end{equation}
Substituting \eqref{px1} in \eqref{pxsum} we get
\begin{eqnarray*}
\lambda^{\frac{1}{q-1}} = \frac{1}{
    \isumx \left[
	r(x) 
	\Bigl(
	1-(1-q)\isumm\beta_m u_m(x)
	\Bigr) ^ {\frac{1}{q-1}}
    \right]
} \id
\end{eqnarray*}
Substituting in \eqref{px1} we get
\begin{eqnarray}
p(x) = \frac{r(x) \Bigl(1-(1-q)
\isumm\beta_mu_m(x)\Bigr)^{\frac{1}{q-1}}}{\iz}\label{px1_1}\ico
\end{eqnarray}
where
\begin{eqnarray}
\iz = \isumx \left[
	r(x) 
	\Bigl(
	1-(1-q)\isumm\beta_m u_m(x)
	\Bigr) ^ {\frac{1}{q-1}}
    \right]\nonumber\id
\end{eqnarray}
equation \eqref{px1_1} can be rewritten as
\begin{align}
p(x) = \frac{r(x)}{\iz
\;\;\exp_q\Big(-\isumm\beta_mu_m(x)\Big)}\label{px1_2}\id
\end{align}
Where  Where $\exp_q$ is exponentiation in $q$-deformed algebra
\cite{Borges:2004:ApossibleDeformedAlgebra}, and is defined as,
\begin{equation*}
\exp_q(x) = \left\{
  \begin{array}{l l}
    [1+(1-q)x]^\frac{1}{1-q} & \quad \text{if } 1+(1-q)x \ge 0\\
    0 & \quad \text{otherwise}\id
  \end{array} \right.
\end{equation*}
using the relation $\frac{1}{\exp_q(x)} =
\exp_q\left(\frac{-x}{1+(1-q)x}\right)$, we get
\begin{equation}
 p(x) = \frac{r(x)}{\iz}
\exp_q\left(\frac{\isumm\beta_mu_m(x)}{1-(1-q)\isumm\beta_mu_m(x)}\right)\id
\end{equation}
Note that we need an extra condition known as \textit{Tsallis cut-off condition}
to prevent negative values for $p(x)$. We have assumed this condition to be
implicit.
\section{Shore and Johnson Properties involving maximum
entropy}\label{snjmini}
In this section we shall discuss properties which depend on the formalism
employed.
\begin{thm}[Uniqueness]
 For $q>0$ given a prior, the posterior probability distribution is unique.
\end{thm}
\begin{proof}
 For $q>0$ Tsallis divergence is a convex function, for both its parameter.
Since the constraint set $\ic$ is a convex set, the minimization is always
unique.
\end{proof}
\begin{thm}[Reflexiveness]\label{res4}
 For $q > 0$, given a prior $r$ and constraint set $\ic$, the posterior
obtained by minimizing the Tsallis divergence is same as $r$ if and
only if $r \in \ic$
\end{thm}
\begin{proof}
 This property follows directly from the following facts $I_q(p||r) = 0
\quad\text{iff}\quad  p = r$ and $I_q(p||r) > 0 \quad\text{for}\quad q > 0$.
\end{proof}
\begin{thm}[Idempotence]
 Given a prior $r$ and constraint set $\ic$, let $p$ be the
posterior obtained, then $\imin{u \in \ic} I_q(u||p) = p$,
i.e., taking the same information into account twice has the same effect as
taking it into account once.
\end{thm}
\begin{proof}
 This is a simple corollary of proposition \ref{res4},  since $p \in \ic$ the
posterior obtained by taking $p$ as prior and $\ic$ as constraint, will also be
$p$.
\end{proof}
\begin{thm}[Invariance]\label{res7}
 Given a prior $r$ consider the constraint sets $\ic_1$ and $\ic_2$,
let $p = \imin{u \in \ic_1} I_q(u||r)$, then following relations hold
\begin{align}
 p &= \imin{u\in \ic_1 \wedge \ic_2}I_q(u||r)\label{res7_1}\\
 &= \imin{u \in \ic_1 \wedge \ic_2} I_q(u||p)\label{res7_2}\\
 &= \imin{u \in \ic_2} I_q(u||p)\label{res7_3}\id
\end{align}

\end{thm}
\begin{proof}
 $p \in \ic_1$ and $p \in \ic_2$ hence $p \in \ic_1 \wedge \ic_2$ so from
proposition \ref{res4} both, \eqref{res7_2} and \eqref{res7_3} follow. We know
that $p = \imin{u \in \ic_1} I_q(u||r)$ and $p \in \ic_1 \wedge
\ic_2$ from the above two, \eqref{res7_3} follows.
\end{proof}
The result shows that if the posterior obtained from $\ic_1$ is an element of
$\ic_2$ then applying $\ic_2$ on the posterior in different ways does not result
in any change.
\begin{thm}[Weak Subset Independence]
 Let $S_1, S_2, \dots, S_n$ be a partition of $\ix$. Let the new information
$\ic$
comprise about each of the conditional densities $p(x/x\in s_i),\; i= 1 \dots
n$. Thus, $\ic = \ic_1 \wedge \ic_2 \wedge \dots \wedge \ic_n$, where $\ic_i$
is the constraint set on the conditional densities of $S_i$.Then given the prior
distribution $r$
\begin{align}
 p_{\ic}^{min}(x/x\in S_i) = \imin{p \in \ic_i}
I_q(p||r_i), \quad q\in(0,1)\label{res3_1}\ico
\end{align}
 and
\begin{align}
 I_q(p_{\ic}^{min}||r) &= \isumn u_i\; I_q(p_i||r_i) - \isumn u_i
\;\ln_q \frac{s_i}{u_i} \nonumber\\
 &\quad+ (1-q) \isumn \left( u_i\; \ln_q \frac{s_i}{u_i}\;
I_q(p_i||r_i)\right)\label{res3_2}\ico
\end{align}
hold where
\begin{align*}
 p_{\ic}^{min} &= \imin{p\; \in \ic} I_q(p||r)\ico\\
 p_i(x) &= p_{\ic}^{min}(x/x \in S_i)\ico\\
 r_i(x) &= r(x/x \in S_i)\id
\end{align*}
$s_i$ are the prior probability of being in each subset,
given by $s_i = \sum_{x \in S_i} r(x)$, and $u_i$ are the posterior probability
of being in each subset, given by $u_i = \sum_{x \in S_i} p_\ic^{min}(x)$.
\end{thm}
\begin{proof}
 Let $\mathcal{R}$ be the information defined by the constraint $\sum_{x \in
S_i} p(x) = u_i$, then it follows from proposition \ref{res7} that
 \begin{equation*}
  \imin{p\; \in \ic} I_q(p||r) = \imin{p\; \in \ic\wedge\mathcal{R}}
I_q(p||r)\id
 \end{equation*}
 Now we can apply proposition \ref{res2} to get \eqref{res3_1} and
\eqref{res3_2}.

\end{proof}
This result is same as proposition \ref{res2} and has the same interpretation.
This difference here lies in the fact that we do not have a prior information
$\mathcal{M}$ regarding the total probability in each subset.
\begin{thm}[Subset Aggregation]
 Let $S_1,S_2,\dots,S_n$ be a partition of $\ix$. Let $\ig$ be a transformation
which converts a given distribution $p$ to discrete distribution over $S_i$,
the transformation is defined by
 \begin{equation*}
  p'(x_i) = \ig p = \int_{S_i} p(x) \mathrm{d}x\ico
 \end{equation*}
 where $x_i$ is  a discrete state corresponding to $x \in S_i$. Let $\ic$' be
the new information about the distribution $\ig p$. Then for a given prior
$r$, then
 \begin{align}
  r (x/x\in S_i) &= p_{min} (x/x\in S_i)\label{res8_1}\ico\\
  \ig p_{min} &= \ig (p_{min})\label{res8_2}\ico\\
  \text {and }
  I_q(\ig p_{min}||\ig r) &= I_q (p_{min}||r)\label{res8_3}\ico
 \end{align}
 where $p_{min} = \imin{p \in \ig^{-1}(\ic')} I_q (p||r)$.
\end{thm}
\begin{proof}
 The constraint set $\ic'$ is defined by a set of expectations 
 \begin{equation*}
  \isumn p'(x_i)u_m(x_i) = \iexpe{u_m} \quad m=1\dots M\id
 \end{equation*}
 In terms of $p = \ig p'$ the constraint set can be represented as
 \begin{equation*}
  \int_{\ix} p(x)w_m(x) = \iexpe{u_m} \quad m=1\dots M\ico
 \end{equation*}
 where $w_m$ is defined as
 \begin{equation*}
  w_m(x) = u_m(x_i), \quad\text{for}\quad x\in S_i,\quad i=1\dots n\ico
 \end{equation*}
 i.e., $w_m$ is constant in each of the subsets $S_i$.\\
 From \eqref{px1_2} we get
 \begin{equation}\label{xyz1}
  p_{min}(x)=\frac{r(x)}{\iz \;\;\exp_q\Big(-\isumm\beta_m w_m(x)\Big)}\id
 \end{equation}
Since $w_m$ is a constant within each subset $S_i$ and $iz$ is a
constant in
itself. So equation \eqref{xyz1} reduces to:
 \begin{align*}
  p_{min}(x) &= K_i\;r(x)\ico
 \end{align*}
 where $K_i$ is a constant for each subset. Now we have
 \begin{align*}
  r (x/x\in S_i) &= r(x)\left/ \int_{y \in S_I} r(y)\right.\\
  &= p_{min} (x/x\in S_i)\id
 \end{align*}
 This proves \eqref{res8_1}.\\
 Now consider the relation
 \begin{align}
  I_q(p_{min}||r) &= \isumn u_i\; I_q(p_i||r_i) - \isumn u_i
\;\ln_q \frac{s_i}{u_i} \nonumber\\
 &\quad+ (1-q) \isumn \left( u_i\; \ln_q \frac{s_i}{u_i}\;
I_q(p_i||r_i)\right) \ico\label{xyz2}
 \end{align}
 which follows from \eqref{res3_2}. where
\begin{align*}
 p_i(x) &= p_{min}(x/x \in S_i)\ico\\
 r_i(x) &= r(x/x \in S_i)\ico\\
 s_i &= \sum_{x \in S_i} r(x)\ico\\
 \text{and }
 u_i &= \sum_{x \in S_i} p_{min}(x)\id
\end{align*}
From \eqref{res8_1} we have that $p_i(x)=r_i(x)$ and hence $I_q(p_i||r_i) = 0$.
Now equation \eqref{xyz2} reduces to 
\begin{align*}
 I_q(p_{min}||r) &= - \isumn u_i\;\ln_q \frac{s_i}{u_i}\nonumber\\
 &= I_q(\ig p_{min}||\ig r)\id
\end{align*}
This proves \eqref{res8_2} and \eqref{res8_3}.

\end{proof}
\section{Some Observations On Duality and Pythagoras}
\subsection{Pythagorean Property}
Because of its extensive use in many problems, Pythagorean property is very
important. It has been shown to exist for both second and third formalisms,
involving $q$-expectation and normalized $q$-expectation respectively. In this
section we have attempted to find the equivalent result for the classical
expectation. The result we got is not promising but we present it here for
future reference, and to introduce an alternative way to manipulate the Lagrange
multipliers. Lets formally state our problem at hand:
\paragraph{Problem statement : }
Let $r$ be the prior distribution and let $p$ be the posterior got by
minimizing the Tsallis divergence subject to the constraint set $\ic$
\begin{align*}
 \isumx p(x) u_m(x) = \iexpe{u_m} \quad m = 1\dots M\id\\
 \intertext{Let $l$ be another distribution satisfying the constraint}
 \isumx l(x) u_m(x) = \iexpe{w_m} \quad m = 1\dots M\id
\end{align*}
We are interested in finding the relation between $\iexpe{u_m}$ and
$\iexpe{w_m}$ so as to minimize the divergence $I_q(l||p)$.
\paragraph{Solution}
To find a solution to this problem we shall minimize the Tsallis divergence in
a different manner. We start the minimization with the following Lagrangian
\begin{eqnarray*}
 \lefteqn{\mathcal{L} =}\\
 &\isumx p(x)\frac{\left[\frac{p(x)}{r(x)}\right]^{q-1}-1}{q-1}
 - (1-q\lambda) \bigl( \isumx{p(x)} - 1\bigr) \\
 &+ \isumm q\beta_m(\isumx{u_m(x) p(x)} - \iexpe{u_m})\ico
\end{eqnarray*}
differentiating $\mathcal{L}$ with respect to $p(x)$  and equating to $0$, we
get
\begin{align}
 \ln_q \left(\frac{r(x)}{p(x)}\right) & = \lambda - \isumm
\beta_mu_m(x)\label{xyz4}\ico\\
 p_{min} &=p(x) = \frac{r(x)}{\exp_q(\lambda - \isumm\beta_mu_m(x))}\id
\end{align}
 Multiplying equation \eqref{xyz4} by $p(x)$ and summing it over $\ix$ we get
 \begin{align*}
  \isumx p(x) \;\ln_q \left(\frac{r(x)}{p(x)}\right) & =\isumx p(x) \lambda\\
&\qquad-\isumx \isumm p(x) \;\beta_mu_m(x)\ico\\
  -I_q^{min}(p||r) &= \lambda -  \isumm \beta_m\iexpe{u_m}\id
 \end{align*}
 Differentiating $I_q^{min}(p||r)$ with respect to $\iexpe{u_m}$ we get
 \begin{equation}
  \idiff{I_q^{min}}{\iexpe{u_m}} = \beta_m\label{leg2_1}\id
 \end{equation}
 Substituting 
 \begin{equation}
  \beta_m = \beta_m' (1+(1-q)\lambda)\ico
 \end{equation}
 equation \eqref{xyz4} reduces to
 \begin{align*}
   \ln_q \left(\frac{r(x)}{p(x)}\right) &=\lambda \oplus_q
-\isumm\beta_mu_m(x)\\
p(x) & = \frac{r(x)}{\iz \;\exp_q(-\isumm\beta_mu_m(x))}\ico
 \end{align*}
where $\iz = \exp_q(\lambda)$.
Hence equation \eqref{xyz4} can be rewritten as
\begin{align*}
 \ln_q \left(\frac{r(x)}{p(x)}\right) & = \ln_q\iz - \isumm
\beta_mu_m(x)\id
\end{align*}
 Multiplying this equation $p(x)$ and summing it over $\ix$ we get
\begin{align*}
 -I_q^{min} = \ln_q\iz - \isumm \beta_m \iexpe{u_m}\id
\end{align*}
 Differentiating $I_q^{min}$ with respect to $\beta_m$ and equating to $0$ we
get
\begin{equation}
 \idiff{\ln_q\iz}{\beta_m} = \iexpe{u_m}\label{leg2_2}\id 
\end{equation}
Equations \eqref{leg2_1} and  \eqref{leg2_2} are the Legendre transform
relations. Given the relations and the divergence minimization let us look at
the Pythagorean property.\\
We want to minimize the divergence $I_q (l||p)$. For this we will proceed as
follows
\begin{align*}
 I_q(l||r) &- I_q(l||p) = -\isumx l(x)\left[ \ln_q \frac{r(x)}{l(x)} - \ln_q
\frac{p(x)}{l(x)}\right]\ico
\intertext{using the relation $\ln_q \left(\frac{x}{y}\right) = y^{q-1}(\ln_qx -
\ln_qy)$, we get}
I_q(l||r) &- I_q(l||p) \\
&= -\isumx l(x)\left[\ln_q \frac{r(x)}{p(x)} \left[1+
(1-q) \ln_q\frac{p(x)}{l(x)}\right]\right]\ico
\end{align*}
using equation \eqref{xyz4}
\begin{align}
& I_q(l||r) - I_q(l||p) \nonumber \\
&= -\isumx l(x)\left[\left(\lambda - \isumm \beta_mu_m(x)\right)\right.\nonumber
\\
&\qquad \qquad \qquad \qquad \left.\left(1+ (1-q)
\ln_q\frac{p(x)}{l(x)}\right)\right]\nonumber\\
&= \lambda - \isumm \beta_m \iexpe{w_m} - (1-q) \lambda \; I_q(l||p)\nonumber\\
&\qquad-(1-q) \isumx \left(l(x) \;\ln_q\frac{p(x)}{l(x)} \; \isumm \beta_m
u_m(x)
\right) \label{xyz6}\id
\end{align}
The minimum of $I_q(l||p)$ is achieved for
\begin{align*}
 \idiff{I_q(l||p)}{\beta_m} = 0\id
\end{align*}
Differentiating \eqref{xyz6} we get
\begin{align*}
 &\idiff{\lambda}{\beta_m} - \iexpe{w_m} - (1-q) I_q(l||p)
\idiff{\lambda}{\beta_m} \\
&\;-(1-q) \isumx l(x) \idiff{}{\beta_m} \left[
\ln_q\frac{p(x)}{l(x)} \; \isumm \beta_m u_m(x)\right] = 0\id
\end{align*}
Using equation \eqref{leg2_2} we get
\begin{align}
 \iexpe{u_m} &- \iexpe{w_m} - \iexpe{u_m}\;I_q(l||p) \nonumber\\ 
 &=(1-q) \isumx l(x) \idiff{}{\beta_m} \left[
\ln_q\frac{p(x)}{l(x)} \; \isumm \beta_m u_m(x)\right]\nonumber\\
&=(1-q) \isumx l(x) \left[ \beta_m \ln_q\frac{p(x)}{l(x)} \right.\nonumber\\
&\qquad \qquad \qquad \left. + \isumm \beta_m u_m(x) \idiff{}{\beta_m} \left(
\ln_q\frac{p(x)}{l(x)}\right) \right]\id
\label{py1}
\end{align}
Evaluating it further we by using the relations $\ln_q\left(\frac{x}{y}\right) =
\frac{\ln_qx - \ln_q y}{1 + (1-q)\ln_qy}$ and $
 p(x) = \frac{r(x)}{\lambda - \isumm \beta_mu_m(x)}$,
We get
\begin{align}
 \iexpe{w_m} &=\nonumber\\
 &\iexpe{u_m}(1 -  I_q(l||p)) + (1-q)\beta_m I_q(l||p)\nonumber\\
 &- \isumx \frac{l^q(x)\;(u_m(x)- \iexpe{u_m}) \;\Psi}{\left[ 1 + (1-q) \ln_q
\frac{r(x)}{p(x)}\right]^2}\ico
\end{align}
where $\Psi = (1+(1-q) \ln_qr(x)) \isumm \beta_mu_m(x)$. Note that in this
expression $r(x)$ can be replaced in terms of $p(x)$.\\
Though this relation does not seem promising, we have mentioned it here for the
sake of completion.\\
\subsection{Additive transformation - $\mathbf{q\leftrightarrow2-q}$}
In $q-$deformed algebra there exists a $q\leftrightarrow2-q$ duality. Which is
the following:
\begin{align}
 \ln_q (1/x) &= \ln_{2-q} (x)\ico\\
 \exp_q(-x) & = \frac{1}{\exp_{2-q} (x)}\label{xyz3}\id
\end{align}
Using this duality Tsallis entropy has been well studied, i.e., various
properties of $S_{2-q}$ has been studied. Initial observations regarding
$S_{2-q}$ were made by Baldovin and Robledo \cite{baldovin:2004:nonextensive}.
Naudts \cite{Naudts:2004:GeneralizedThermostatisticsAndMeanfieldTheory} has
further analyzed both
the dualities. More study has been carried forward by Wada and
Scarfone \cite{WadaScarfone:2005:ConnectionsBetweenTsallisFormalismEtc}. they
have found relations between the Lagrange multipliers of both the dualities. In
this section we introduce a similar transformation for Tsallis divergence.\\
Given a prior $r$ and the constraints set $\ic$ defined by
\begin{align}
\isumx{p(x)} &= 1\nonumber\ico\\
p(x) &\ge 0\nonumber\ico\\
\isumx{u_m(x) p(x)} &= \iexpe{u_m}, \quad m = 1,\dots,M\nonumber\id
\end{align}
from equation \eqref{px1_2} we have
\begin{equation*}
 p(x) = \frac{r(x)}{\iz \;\;\exp_q\Big(-\isumm\beta_mu_m(x)\Big)}\ico
\end{equation*}
and using the relation \eqref{xyz3} it becomes
\begin{equation*}
 p(x) = \frac{r(x)\;\;\exp_{2-q}\Big(\isumm\beta_mu_m(x)\Big)}{\iz }\id
\end{equation*}
This form for the posterior is very good and is the basis for the
$q\leftrightarrow2-q$ transformation. Note that 
\begin{equation*}
 2 - (2-q) = q\ico
\end{equation*}
i.e if we minimize $I_{2-q} (p||r)$ instead of $I_q (p||r)$, we have.
\begin{align*}
 p(x) &= \imin{p \in \ic} I_{2-q} (p||r)\\
 &=\frac{r(x)\;\;\exp_{q}\Big(\isumm\beta_mu_m(x)\Big)}{\iz }\id
\end{align*}
\section{Conclusion}
In this work we explored Shore and Johnson properties for Tsallis formalism of
the third kind involving normalized $q$-expectation, it was observed that none
of these properties hold for the formalism. Whereas in the study of first
formalism involving classical expectation, we have been able to establish
substantial number of Shore and Johnson properties. We were also able to
establish a crude form of Pythagorean relation. We have also been found a  $q
\leftrightarrow 2-q$ additive transformation, which gives a very good form for
the posterior distribution. We conclude from  these observations that the first
formalism is of stronger theoretical and practical significance; and these
results along with the $q \leftrightarrow 2-q$ additive transformation also
provides some ground work for definition of a power law family.
\bibliography{1}
\end{document}